\documentclass[letterpaper, 10 pt, conference]{ieeeconf}
\IEEEoverridecommandlockouts
\usepackage{graphicx}
\usepackage{subcaption}								
\usepackage{amsfonts}										
\usepackage{amssymb}
\usepackage{algorithm}

\usepackage{mathtools}										

\usepackage{amscd}										
\usepackage{booktabs}										
\usepackage{graphics} 										
\usepackage{epsfig} 										
\usepackage{contour}										
\usepackage{bm}
\usepackage{times} 											
\usepackage{rotating}										
\usepackage{cite} 									
\usepackage{listings}										
\usepackage{xcolor} 	

\DeclareMathAlphabet{\mathpzc}{OT1}{pzc}{m}{it}
\DeclareMathAlphabet{\mathcal}{OMS}{cmsy}{m}{n}

\newtheorem{corollary}{Corollary}                   						
\newtheorem{remark}{Remark}
\newtheorem{theorem}{Theorem}
\newtheorem{assumption}{Assumption}

\newtheorem{definition}{Definition}

\newtheorem{objective}{Objective}

\usepackage{cases}
\usepackage{tikz}				
\usetikzlibrary{arrows,shapes,backgrounds,calc,positioning,patterns}
\usepackage{balance}

\DeclareMathOperator{\diag}{diag}

\usepackage{array}

\usepackage{fixltx2e}

\usepackage{float}
\usepackage{url}

\newcommand{\setof}[1]{\left\{#1\right\}} 

\newcommand{\brackets}[1]{\left[#1\right]}
\newcommand{\braces}[1]{\left(#1\right)}

\newcommand{\norm}[1]{\|#1\|}

\newcommand{\de}{\dot{e}}
\newcommand{\dde}{\ddot{e}}

\usepackage{pifont}
\newcommand{\cmark}{\ding{51}}%
\newcommand{\xmark}{\ding{55}}%

\newtheorem{prop}{Proposition}

\begin{document}
	%
	\title{A novel passivity based controller for a piezoelectric beam}
	
	\author{
	Krishna~C.~Kosaraju,	Matthijs C. de Jong, 
		 and
		Jacquelien~M.A.~Scherpen
		\thanks{M.C. de Jong, K.C. Kosaraju, and J.M.A. Scherpen are with Jan C. Wilems Center for Systems and Control, ENTEG, Faculty of Science and Engineering, University of Groningen, Nijenborgh 4, 9747 AG Groningen, the Netherlands (email: \{k.c.kosaraju, m.c.dejong, j.m.a.scherpen\}@rug.nl).}
		\thanks{This work is supported by the Netherlands Organization for Scientific Research through Research Programme ENBARK+ under Project 408.urs+.16.005.}%
	}


	\maketitle
	\begin{abstract}
This paper presents a new passivity property for distributed piezoelectric devices with integrable port-variables. We present two new control methodologies by exploiting the integrability property of the port-variables. The derived controllers have a Proportional-Integral (PI) like structure. Finally, we present the simulation results and an in-depth analysis on the tuning gains explaining their transient and the steady-state behaviours.
\end{abstract}
	
	\IEEEpeerreviewmaketitle
	%
	%
	%
	%
	\section{Introduction}
	A piezoelectric beam is a bar of piezoelectric material, that is much longer than it is wide or thick. 
	By sandwiching the piezoelectric material between two electrode layers, the piezoelectric properties can be exploited. 
	From the perspective of control we are interested in deforming the beam, where the inverse piezoelectric effect is exploited. That is, by means of electric stimuli, injected stresses in the material allow the piezoelectric beam to deform. The orientation of the piezoelectric material determines in what direction(s) the beam deforms. Here it's assumed that the piezoelectric beam shows only longitudinal elongating and compressing behavior, see Fig \ref{fig:piezobeam}. This property is useful in shape control of complex mechanical structures, see for instance \cite{Voss2011Shapecontrol,Voss09disc}.
	Piezoelectric beam models originate from the Maxwell's equations and continuum mechanics, that characterize the electric and mechanical characteristics of the beam, respectively. The assumptions on the electric characteristics can be a static electric field, which is the most commonly used assumption, the quasi-static electric field, which restricts the electric displacement, and the fully dynamic electromagnetic field, which includes the magnetic kinetic effects, see for instance \cite{MenOSIAM2014,voss2010port}.
	For the mechanical characteristics various beam theories can be chosen, such as the Euler-Bernoulli or Timoshenko beam theory, see for instance \cite{carrera2011beam}. For the piezoelectric beam considered in this paper, all relations are assumed linear and for the electric and mechanical characteristics the static electric field with Euler-Bernoulli beam theory \cite{carrera2011beam} are considered. Often, physical systems (including a piezoelectric beam) dissipates energy, by means of some sort of damping.
	\begin{figure}
		\includegraphics[width=\columnwidth]{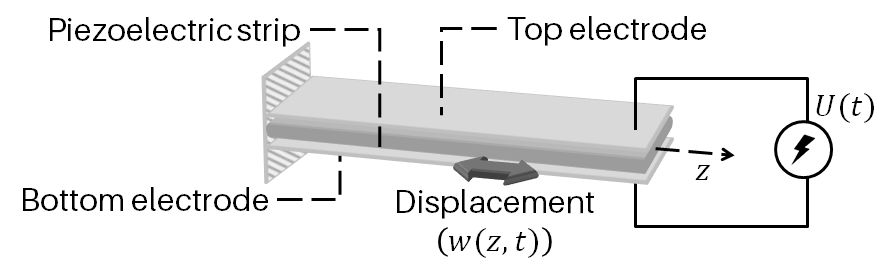}
		\caption{Piezoelectric beam.} 
		\label{fig:piezobeam}
		\vspace{-0.3mm}
	\end{figure}
	In this paper, we start with the distributed port-Hamiltonian formulation \cite{VANDERSCHAFT2002166} for the considered piezioelectric beam.
	The use of the port-Hamiltonian framework \cite{port-HamiltonianIntroductory14} is motivated by its openness to interact with its environment  through the boundary. This allows straightforward interconnection with other open systems, such as controllers and mechanical structures, where the latter is of interest for ongoing research. Furthermore, applying this structure to the system help us discover new passivity properties, that is otherwise not evident. Moreover, this helps us in using the structure preserving spatial discretization method presented in \cite{GoloSchaft2004} for presenting simulation results. In literature, there are several methods that are proposed for stability-analysis/control of a general (damped) wave equations, to cite a few see \cite{shubov1996basis,liu2009elementary,cox1995rate} and the reference therein. This paper propose two alternative approaches using passivity theory. The following are the main contributions.
\vspace{-2mm}
	\subsection{Contributions}
	In this paper, we propose a novel passivity based control technique for a piezoelectric beam actuated through the boundary by an applied voltage $U(t)$, that has viscous damping. It is well known that boundary control of these systems are constrained by the {\em dissipation obstacle} like the problems in \cite{915398}. A work around for this problem is recently presented in \cite{Macchelli2017_BCLDPS_dissipation} using shifted passivity properties \cite{shifted_passivity}. On the other hand, we propose an alternative methodology, which has the following key contributions:
	\begin{itemize}
	    \item [(i)] The storage function considered, is a function of velocities (rather than states). This results in a new passivity property whose port-variable have integrability properties, inspired by \cite{tac, NOLCOS, 7846443,mtns,Cucuzzella_arxiv2019}.
	    \item[(ii)] By utilizing these integrability properties, we propose two control methodologies. These techniques are  constructive and do not rely on finding Casimir functions as in \cite{SchSiu}.
	    \item[(iii)] The proposed control techniques have a Proportional-Integral (PI)  structure. We present an in depth analysis on the effect of tuning parameters on the transient, steady-state responses.
	    \vspace{-2mm}
	\end{itemize}
	\subsection{Outline}
	The paper is organized as follows. In Section II, we present the model of a piezoelectric beam and its port-Hamiltonian formulation. Further, in section III, we formulate the boundary control problem of the considered piezoelectric beam and state the required assumptions on the knowledge of system parameters and measured states. In section IV, we present a new passivity property and propose two control methodologies. Finally, in section V, we present the simulation results and give some analysis on gain tuning of the proposed techniques. Further, in Section VI, we give some concluding remarks and possible future directions for the presented work.
	\subsection{Notation}
		Denote the spatial domain $Z=[0,1]\subset \mathbb{R}$, and $\partial Z=\{0,1\}$ represents the boundary of $Z$.  The operator $\mathrm{d}$ denotes the partial derivative with respect to $z$, i.e., $\mathrm{d}=\dfrac{\partial }{\partial z}$. $\mathcal{L}^2(0,1)$ denotes the space of square-integrable functions on domain $z\in(0,1)$ and $H^1(0,1)$ denotes the first order Sobolov space. Let $z\in Z$ and $t\in \mathbb{R}_{+}$ denotes the spatial and temporal independent variables. Consider the functions $\alpha(z,t),\beta(z,t)$ mapping from $[0,1]\times \mathbb{R}_{+}\rightarrow \mathbb{R}$ that are twice differentiable. For brevity, we denote $\alpha_z=\dfrac{\partial \alpha}{\partial z}$, $\alpha_{zz}=\dfrac{\partial^2 \alpha}{\partial z^2}$, $\dot\alpha=\dfrac{\partial \alpha}{\partial t}$, $\ddot\alpha=\dfrac{\partial^2 \alpha}{\partial t^2}$, $\alpha(0)=\alpha(0,t)$, $\alpha(1)=\alpha(1,t)$ and $\alpha^\ast =\lim\limits_{t\rightarrow \infty}\alpha(t,z)$, i.e., $\alpha^{\ast}$ represents the steady-state value of $\alpha$. Consider a functional $\mathcal{H}(\alpha,\beta)=\int_0^1\mathrm{H}(\alpha,\beta) dz$, where $\mathrm{H}(\alpha,\beta) :\mathcal{L}^2(0,1) \times \mathcal{L}^2(0,1)  \rightarrow\mathbb{R}$. Then $\delta_{\alpha}\mathcal{H}$, $\delta_{\beta}\mathcal{H}$ denote the variational derivative of $\mathcal{H}$ with respect to $\alpha$ and $\beta$, respectively, given by
		\vspace{-3mm} 
		\begin{eqnarray*}
			\delta_{\alpha}\mathcal{H}= \dfrac{\partial }{\partial z}\dfrac{\partial \mathrm{H}}{\partial \alpha_z}-\dfrac{\partial \mathrm{H}}{\partial \alpha},~
			\delta_{\beta}\mathcal{H}= \dfrac{\partial }{\partial z}\dfrac{\partial \mathrm{H}}{\partial \beta_z}-\dfrac{\partial \mathrm{H}}{\partial \beta}.
		\end{eqnarray*}
		Let $a,b \in \mathbb{R}$, then $\diag\{a,b\}\in \mathbb{R}^{2\times 2}$ represent a diagonal matrix with $a$ and $b$ as its diagonal entries.
	\section{Model}
	Consider a viscously damped linear piezoelectric beam with static electric field and Euler-Bernoulli beam theory of length $\ell=1$, such that the spatial variable $z\in Z$. Let the longitudinal displacement with respect to the initial position be denoted by $w(z,t)$ and its velocity is denoted by $\dot{w}(z,t)$. The strain (or deformation) of the beam is denoted by $w_z(z,t)$. Assume that the stress in the beam is linearly related through Hooke's law with the strain of the beam, by the stiffness $C$. Let $\rho$ denote the mass-density and $b$ the viscous damping coefficient. From first principles and application of Hamilton's principle \cite{lanczos1970variational}, the partial-differential-equation (PDE)
	\begin{align}\label{eq:pde_piezo_static}
		\begin{split}
		\rho \ddot{w}(z,t)&=Cw_{zz}(z,t)-b\dot{w}(z,t),
		\end{split}
	\end{align}
	 is obtained, where the coupling between the mechanical and electric domain is provided through the piezoelectric coupling coefficient $\gamma$ that converts actuation by means of an applied voltage $U(t)$ into a mechanical force acting through the boundary. Together, with a clamped left side of the beam, the dynamical equation \eqref{eq:pde_piezo_static} with boundary conditions 
	\begin{align}\label{eq:pde_piezo_static_bc}
	\begin{split}
    w(0,t)&=0, ~ \dot{w}(0,t)=0\\
	Cw_z(1,t)&=-\gamma U(t),
	\end{split}
	\end{align}
	describes the inverse piezoelectric effect of a piezoelectric cantilever beam, that can be stabilized by an appropriate applied voltage $U(t)$ as will be done in this paper. For simplicity of interpretations, from now on the temporal and spatial dependencies of the variables are omitted. Unless special care requires otherwise. Next, we present the distributed port-Hamiltonian formulation \cite{VANDERSCHAFT2002166} for the piezoelectric beam. 
	\subsubsection{Port-Hamiltonian formulation}
 Let $\alpha=\left(\alpha_q,\alpha_p\right)^\top\in {\mathcal{L}^2(0,1)\times \mathcal{L}^2(0,1)}$ denote the energy variables, given by $\alpha_q=w_z,~\alpha_p=\rho \dot{w}$, that describe the total energy of the system \eqref{eq:pde_piezo_static}
 \vspace{-2mm}
	\begin{eqnarray}
	\mathcal{H}=\frac{1}{2}\int_0^1\left(\alpha^\top\mathrm{ P} \alpha\right)dz
	\end{eqnarray}
	\vspace{-1mm}
	where $\mathrm{ P}=\diag\{C,\dfrac{1}{\rho}\}$. Let $H_0^1(0,1):=\setof{z\in H^1(0,1)|z(0)=0}.$ Consequently, we denote the flow and effort variables as $f =\left(f_p,f_q\right)^\top\in \mathcal{L}^2(0,1)\times \mathcal{L}^2(0,1)$ and  $e=\left(e_q,e_p\right)^\top \in {H}^1(0,1)\times {H}_0^1(0,1)$ respectively, given by
	\begin{eqnarray}
	\begin{bmatrix}		f_q\\f_p 	\end{bmatrix}=\dfrac{d}{dt}\begin{bmatrix}		\alpha_q\\\alpha_p 	\end{bmatrix},~~\begin{bmatrix}		e_q\\e_p 	\end{bmatrix}=\begin{bmatrix}		\delta_{\alpha_q}\mathcal{H}\\\delta_{\alpha_p}\mathcal{H} 	\end{bmatrix}.
	\end{eqnarray}
	Denote $\mathrm{d}=\dfrac{\partial}{\partial z}$. Then, system \eqref{eq:pde_piezo_static} can be written in pH formulation, as follows:
	\begin{eqnarray}\label{eq:ph_beam_eqn}
	\begin{bmatrix}		f_q\\f_p 	\end{bmatrix} =
	\underbrace{\begin{bmatrix} 0 & \mathrm{d} \\ \mathrm{d} & -b  \end{bmatrix}}_\mathcal{A}\begin{bmatrix}	e_q \\ e_p 	\end{bmatrix},
	\begin{matrix}
	e_p(1)&=&f_b, \\e_q(1)&=&e_b, 
	\end{matrix}
	\end{eqnarray}
	where $f_b$ and $e_b$ represent the boundary flow and effort variables, respectively, given by $f_b=\dot{w}(1,t)$ and $e_b=-\gamma U(t)$. The domain of $\mathcal{A}$ is $D(\mathcal{A})=\setof{e\in {H}^1(0,1)\times {H}_0^1(0,1)}.$ Next, we establish the following lemma using the Hamiltonian $\mathcal{H}$ as the storage function \cite{VANDERSCHAFT2002166}:
	\begin{corollary}
		The system of equation \eqref{eq:ph_beam_eqn}  is passive with port-variables $e_b=-\gamma U(t)$ and $f_b=\dot{w}(1,t)$.$\hfill\square $
	\end{corollary}
	We finally conclude this section, by rewriting the port-Hamiltonian system \eqref{eq:ph_beam_eqn} using effort variables $e_q, e_p$, given by
	\vspace{-2mm}
	\begin{eqnarray}\label{eq:model_efforts}
	\dfrac{d}{dt}\begin{bmatrix}
	 \dfrac{1}{C} e_q,\\ \rho e_p
	\end{bmatrix}= \begin{bmatrix} 0 & \mathrm{d} \\ \mathrm{d} & -b  \end{bmatrix}\begin{bmatrix}
	 e_q,\\  e_p
	\end{bmatrix},
	\begin{matrix}
	e_p(1)&=&f_b,\\e_q(1)&=&e_b. 
	\end{matrix}
	\end{eqnarray}

	\section{Problem formulation}
	The objective of this paper is to present a passivity based control methodology for the piezoelectric beam \eqref{eq:pde_piezo_static} though boundary \eqref{eq:pde_piezo_static_bc}. We next make it explicit as follows: 
	\begin{objective}\label{obj:strain_z1}
		Stabilize the strain $w_z$ at $z=1$ to a non-zero value $\mathcal{E}^*$, i.e.,
		\vspace{-2mm}
		\begin{eqnarray}
		\lim\limits_{t\rightarrow \infty}w_z(1)\rightarrow \mathcal{E}^\ast
		\end{eqnarray}
		\vspace{-2mm}
$\hfill\square $	\end{objective}
 In the sequel, assume that there exist a steady state solution $(e_q^*,e_p^*,U^*)$ for system \eqref{eq:ph_beam_eqn}. Then, Objective \ref{obj:strain_z1} results in the following boundary value problem
	\vspace{-2mm}
	\begin{align}\label{eq:equilibrium_poi}
	\begin{split}
	\mathrm{d}e_p^*&=0,~ \mathrm{d}e_q^*-be_p^*=0,\hspace{.51cm}
	e_q^*(1)=C\mathcal{E}^\ast,\\
	e_p^\ast(0)&=e_p^\ast(1)=0,\hspace{.5cm}
	U^*=\frac{-C\mathcal{E}^*}{\gamma}.
	\end{split}	\vspace{-2mm}
	\end{align}
	Equation \eqref{eq:equilibrium_poi} indicates that the  control Objective \ref{obj:strain_z1} result in a constant strain type equilibrium (i.e., $w_z=0$).
	In the following assumptions, we state the existence of the steady-state solutions and the available information more explicit:\vspace{-1mm}
	\begin{assumption}[existence of a unique steady-state solution]\label{ass:existance}
	Assume that there exist an unique solution for the boundary value problem \eqref{eq:equilibrium_poi}.$\hfill\square $
	\end{assumption}
	Concisely, the equilibrium point of interest can be written as
	\begin{align}\label{eq:desired_equilibrium}
	\begin{split}\vspace{-2mm}
	\left(e_q^*(1),e_p^*(1),U^*\right)=(C\mathcal{E}^*,0,\frac{-C\mathcal{E}^*}{\gamma}).\vspace{-2mm}
	\end{split}\vspace{-4mm}
	\end{align}
	We next assume the knowledge of required states and parameters for the controller design.
	\begin{assumption}[available information]\label{ass:available}
		The boundary variable $e_p(1)$ is measurable. The system parameters $\rho$, $C$ and $b$ are (possibly unknown) constants. The piezoelectric coupling coefficient $\gamma$ is a known constant. The desired strain $\mathcal{E}^\ast$ and voltage $U^\ast$ are known.$\hfill\square $
	\end{assumption}
\vspace{-1mm}
	\section{Proposed Solution: passivity based control}
	In this section, we first present a new passivity property for the piezoelectric beam model \eqref{eq:ph_beam_eqn}. Later on, we use this passive map and propose two control techniques.
	
	\subsection{New passivity property}
	It is well known that to propose a new passivity property one needs to find a new storage function. Recently, the authors in \cite{tac, NOLCOS, 7846443,mtns,Cucuzzella_arxiv2019} have proposed a new storage function which depends on the `first time derivative of the effort variables'. The use of such storage function yields a new passive map with supply-rate depending on the state-variables and the time-derivative of the state-variables and the input'.  For the piezoelectric beam this results in the following storage function
	\begin{align}\label{eq:storagefunc_plant}
	\begin{split}
	S:=\tfrac{1}{2}\int_{0}^{1}\brackets{\tfrac{1}{C}\dot{e}_q^2+\rho \dot{e}_p^2}dz.
	\end{split}
	\end{align}
	As a consequence, we have the following passivity property.
	\begin{prop}\label{prop:passivity}
		The piezoelectric beam \eqref{eq:model_efforts} is passive with storage function \eqref{eq:storagefunc_plant} and port variables $\dot{U}$ and $-\gamma\dot{e}_p(1)$.
	\end{prop}
	\begin{proof}
		The rate of change of storage function \eqref{eq:storagefunc_plant} along the trajectories of the  piezoelectric beam \eqref{eq:model_efforts} is
		\begin{align*}
		\begin{split}
		\frac{d}{dt}S&=\int_{0}^{1}\brackets{\de_q\tfrac{1}{C}\dde_q+\de_p\rho \dde_p}dz
		\leq -\gamma\de_p|_{z=1} \dot{U}(t)
		\end{split}
		\end{align*}
		Where in the second line we use the spatial dynamics from equation \eqref{eq:pde_piezo_static}. In the fourth and fifth line we use the boundary conditions from equation \eqref{eq:pde_piezo_static_bc}.
		This concludes the proof.
	\end{proof}
	\begin{remark}\label{rem::ext_dyn}
		Note that the storage function is in terms of velocities. Hence, the passivity property presented, not only considers the system dynamics but also considers the differentially extended dynamics of the system \eqref{eq:model_efforts}. This implies, that the state variables in the extended system are $\left(e_q,e_p,\dot{e}_q,\dot{e}_p\right)$ which are useful for the stability analysis. 
	\end{remark}
	\begin{remark}\label{rem_storage_edot_and_e}
		The storage function $S$ in equation \eqref{eq:storagefunc_plant} depends on $\dot{e}_q$ and $\dot{e}_p$. However, $S$ also depends on $e_q$ and $e_p$ through $\dot{e}_q$ and $\dot{e}_p$, see \eqref{eq:model_efforts}. Hence, $S$ can be written as follows \cite{Cucuzzella_arxiv2019},
			\begin{align}\label{eq:storagefunc_planta}
	\begin{split}
	S=\hspace{-1mm}\tfrac{1}{4}\int_{0}^{1}\hspace{-1mm}\brackets{\tfrac{1}{C}\dot{e}_q^2+\rho \dot{e}_p^2\hspace{-1mm}+\hspace{-1mm}\dfrac{1}{\rho} \left(de_q-be_p\right)^2\hspace{-1mm}+\hspace{-1mm}C\left(de_p\right)^2}\hspace{-1mm}dz.
	\end{split}
	\end{align}$\hfill\square $
		\end{remark}
Following the control techniques proposed in \cite{tac}, we present two indirect passivity based PI controllers for the considered Control objective \ref{obj:strain_z1} of the piezoelectric beam \eqref{eq:ph_beam_eqn}. The first technique is called {\em output shaping} which is partially known in the literature as PID-PBC. In output shaping we use the integrability property of the output in shaping the closed-loop storage function thence avoiding the search for Casimir functions. From Proposition \ref{prop:passivity}, one can note that the input port-variable $\dot{U}$ is also integrable. Motivated by this we present a novel control technique called {\em input shaping} for infinite dimensional systems, see \cite{tac} for input shaping on finite dimensional systems. 
	However, before presenting these techniques, we first formulate the stability definitions for distributed parameter system from the literature.
	%
	%
	\subsection{Stability}\label{sec::stability}
	For finite dimensional systems, to prove Lyapunov stability, it suffices to show that the Lyapunov function is positive definite and the first time derivative along the trajectories is negative definite. Additionally, for infinite dimensional system
	the norm associated with the stability argument needs to be specified. Since, stability with respect to one norm does not immediately imply stability with respect to another norm.
	Next the stability arguments and sufficient conditions for stability are defined. Consider an $n$-dimensional smooth manifold $Z$ with smooth $(n-1)$-dimensional boundary $\partial Z$.
boundary
	Let $\mathcal{M}_{\infty}$ be the configuration space of a distributed parameter system with spatial domain $z$, and $\norm{\cdot}$ be a norm on $\mathcal{M}_{\infty}$.
		\begin{definition}\label{def::stability}\cite{ZhenBaoOmer}
			Denote by  $\;x^{\ast}(z)\in \mathcal{M}_{\infty}$ an equilibrium configuration for a distributed parameter system on $\mathcal{M}_{\infty}$. Then, $x^\ast(z)$ is said to be stable in the sense of Lyapunov with respect to the norm $\norm{\cdot}$ if, for every $\epsilon\geq 0$ there exists a $\delta\geq 0$ such that, 
			\begin{eqnarray*}
			\begin{matrix}
				\norm{x(z,0)-x^\ast(z)}\leq \delta & \implies & \norm{x(z,t)-x^\ast(z)}\leq\epsilon
			\end{matrix}
			\end{eqnarray*}
			for all $t\geq 0$, where $x(z,0)\in \mathcal{M}_\infty$ is the initial configuration of the system.$\hfill\square $
		\end{definition}
		In case of the piezoelectric beam \eqref{eq:ph_beam_eqn} the configuration manifold $M_{\infty}$ is $({H}^1(0,1)\times {H}_0^1(0,1))^2$ and the state is the collection of variables $x(z,t)=\left(e_q,e_p,\dot{e}_q,\dot{e}_p\right)$ (see Remark \ref{rem::ext_dyn}). We state the following stability theorem for infinite-dimensional systems, which is also referred to as Arnold’s theorem for stability of infinite-dimensional systems.
			\begin{figure}[ht]
			\begin{subfigure}[b]{0.32\columnwidth}
				\centering
				\includegraphics[width=\columnwidth]{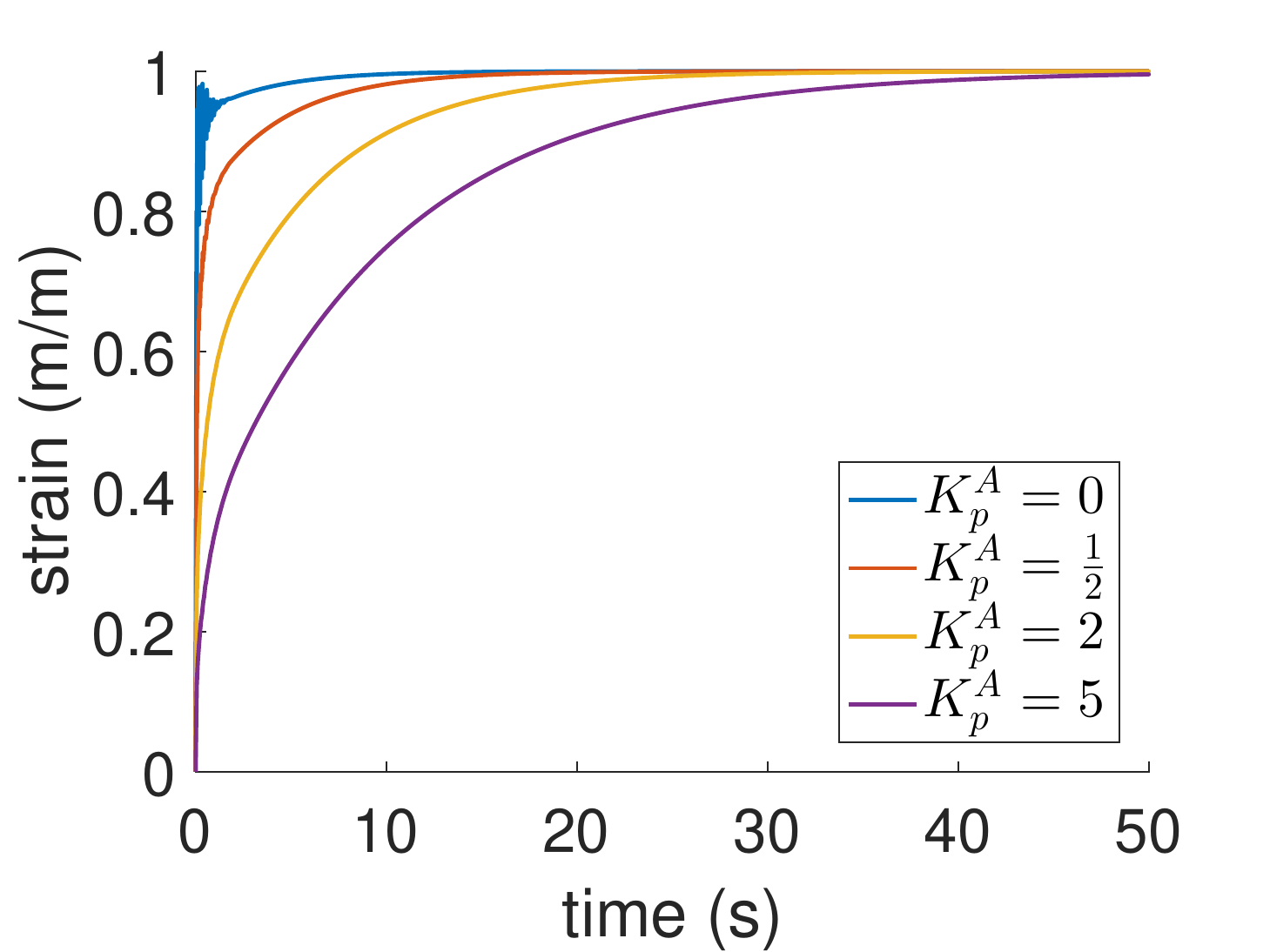}
				\caption{\tiny{$K_i^A=0$, $\bar{U}^*=\tfrac{3}{4}$}}
			\end{subfigure}
			\begin{subfigure}[b]{0.32\columnwidth}
				\includegraphics[width=\columnwidth]{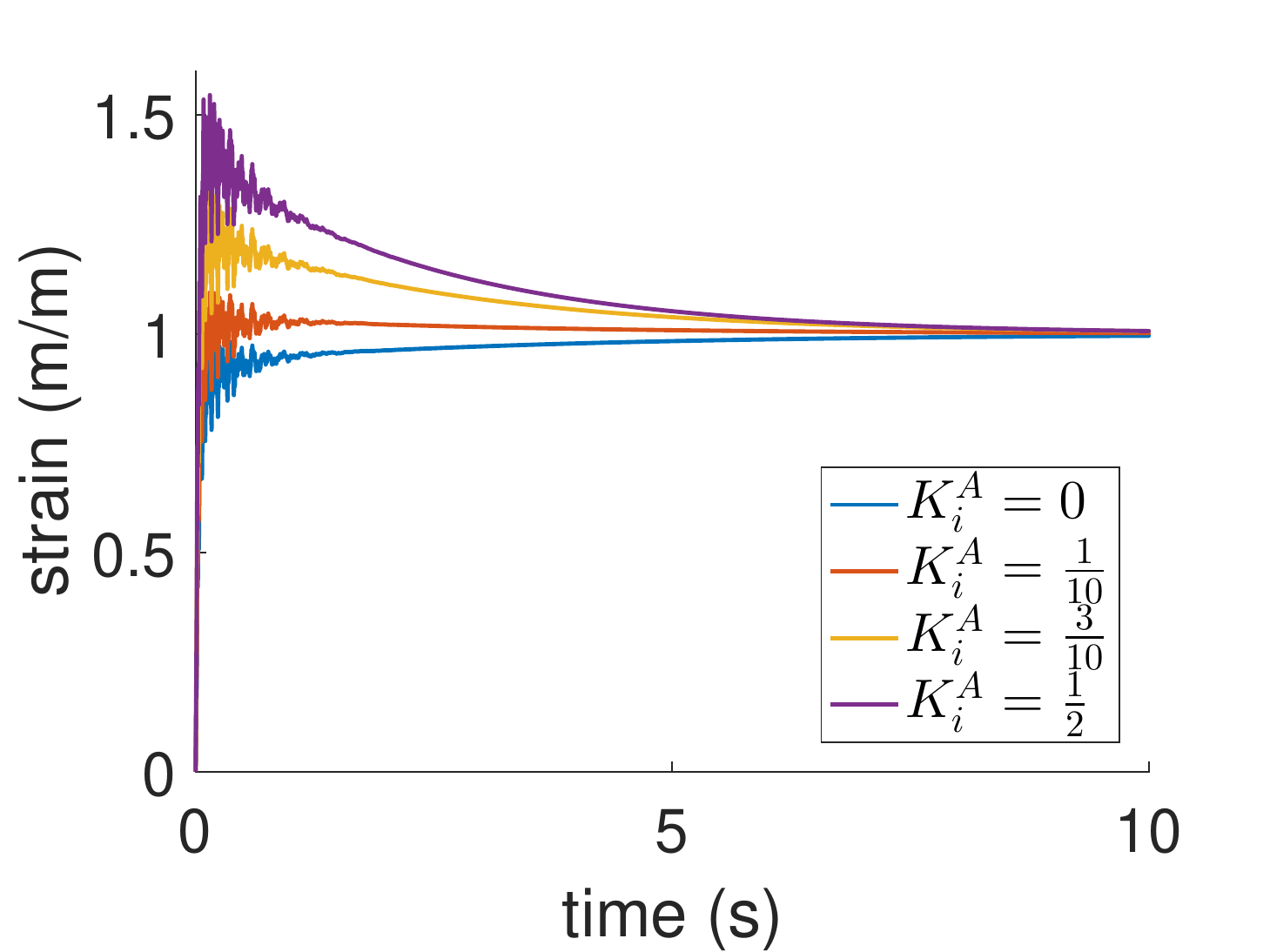}
				\caption{\tiny{$K_p^A=0$, $\bar{U}^*=\tfrac{3}{4}$}}
			\end{subfigure}
			\begin{subfigure}[b]{0.32\columnwidth}
				\includegraphics[width=\columnwidth]{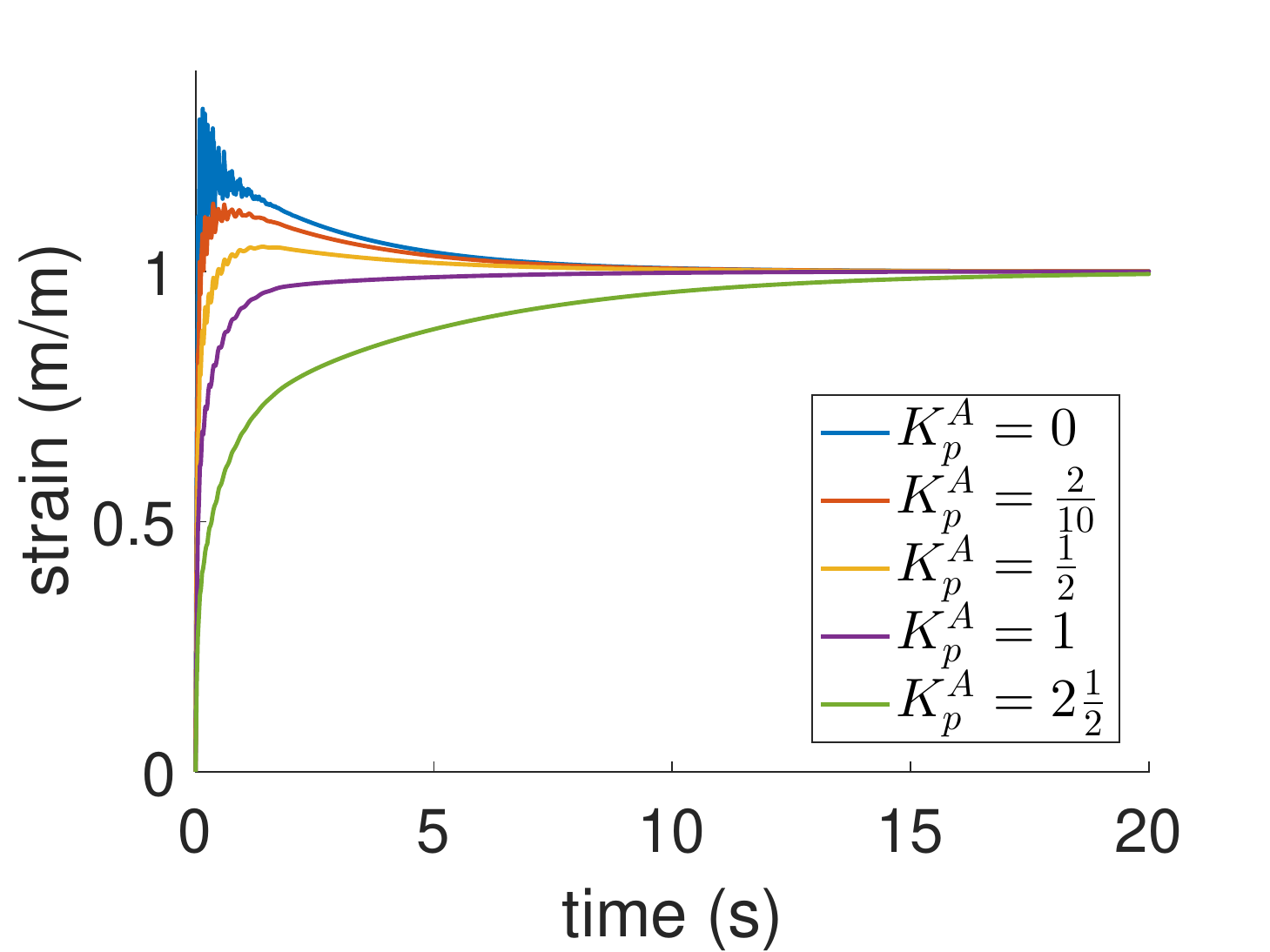}
				\caption{\tiny{$K_i^A=\tfrac{3}{10}$, $\bar{U}^*=\tfrac{3}{4}$}}
			\end{subfigure}
			\caption{Closed-loop response through the output shaping. - In (a) the proportional action is shown, in (b) the integral-term. In (c) the proportional action varies for a fixed integral action wit reference action. The used parameters are given in Table \ref{tab:params_method_A}.}
		%
			\label{fig:sim_method_A}
		\end{figure}
		
		\begin{theorem}\label{thm::stab}\textit{(Stability of an infinite-dimensional system \cite{swaters}):} Consider a dynamical system
		\begin{eqnarray}
		\frac{\partial}{\partial t}x(z,t)=f(x(z,t))
		\end{eqnarray}
 on a linear space $\mathcal{M}_\infty$, where $x^{\ast}(z)\in \mathcal{M}_{\infty}$ is an equilibrium. Assume there exists a solution to the system and suppose there exists function $V:\mathcal{M}_\infty\rightarrow \mathbb{R}$ such that
			\begin{eqnarray}\label{stab_th_firstcond}
			\begin{matrix}
				\delta_{x}V(x^\ast(z))=0 & \text{and} & \frac{\partial V}{\partial t} \leq 0.
			\end{matrix}
			\end{eqnarray}
			Denote $\Delta x=x(z,t)-x^\ast(z)$ and $\mathcal{N}(\Delta x)= V(x^\ast(z) +\Delta x)-V(x^\ast(z))$. Show that there exist a positive triplet $\alpha$, $\gamma_1$ and $\gamma_2$ satisfying 
			\begin{eqnarray}\label{stabthm_cond2}
			\gamma_1\norm{\Delta x}^2 \leq  \mathcal{N}(\Delta x) \leq \gamma_2\norm{\Delta x}^\alpha.
			\end{eqnarray} 
			Then $x^\ast $ is a stable equilibrium.$\hfill\square $
		\end{theorem}
		\subsection{Passivity based control: Output shaping}\label{sec:output_shaping}
		In this method, we construct the closed-loop storage function by utilizing the integrability property of the output port-variable $-\gamma\dot{e}_p(1,t)$. This is similar to the recently emerging PBC techniques called Proportional-Integral-Derivative-PBC (PID-PBC) \cite{7993047}.
		\begin{prop}\label{prop:controla}
			Let Assumption \ref{ass:existance} and \ref{ass:available} hold and consider the piezoelectric beam \eqref{eq:model_efforts}, controlled through the boundary by
			\begin{eqnarray}\label{eq:stabalizing_controller_B}
			\begin{split}
			\dot{\phi}&= e_p(1,t)-e_p^*(1)\\
			U(t)&=
			\frac{1}{\gamma}\brackets{K_i\phi+K_p\braces{e_p(1,t)-e_p^*(1)}}+U^\ast,
			\end{split}
			\end{eqnarray}
			where $K_i,~K_p>0$ are tuning parameters, $\phi \in \mathbb{R}$. Then, the closed-loop system is stable at the operating point \eqref{eq:desired_equilibrium}.
		\end{prop}
		\begin{proof}
			First the integrated output port-variable $e_p(1)$ to shape the closed-loop storage function is used as follows
			\begin{eqnarray}\label{eq:storage_fun_os}
			S_d:=S+\frac{1}{2}K_i\braces{e_p(1)-e_p^*}^2.
			\end{eqnarray}
			As argued in Remark \ref{rem::ext_dyn}, take as state-variable  $x=\left(e_q,e_p,\dot{e}_q,\dot{e}_p\right)^\top$. Moreover, the desired values $x^\ast$ is the solution to the boundary value problem \eqref{eq:equilibrium_poi}. Consider the time derivative of the proposed controller $U(t)$ (in the second line of equation \eqref{eq:stabalizing_controller_B})
			\begin{eqnarray}\label{os_diff_contr}
			\gamma \dot{U}=K_i\left( e_p(1,t)-e_p^*(1)\right)+K_p\dot{e}_p(1,t).
			\end{eqnarray}
									\begin{figure}[ht]
				\begin{subfigure}[b]{0.32\columnwidth}
					\includegraphics[width=\columnwidth]{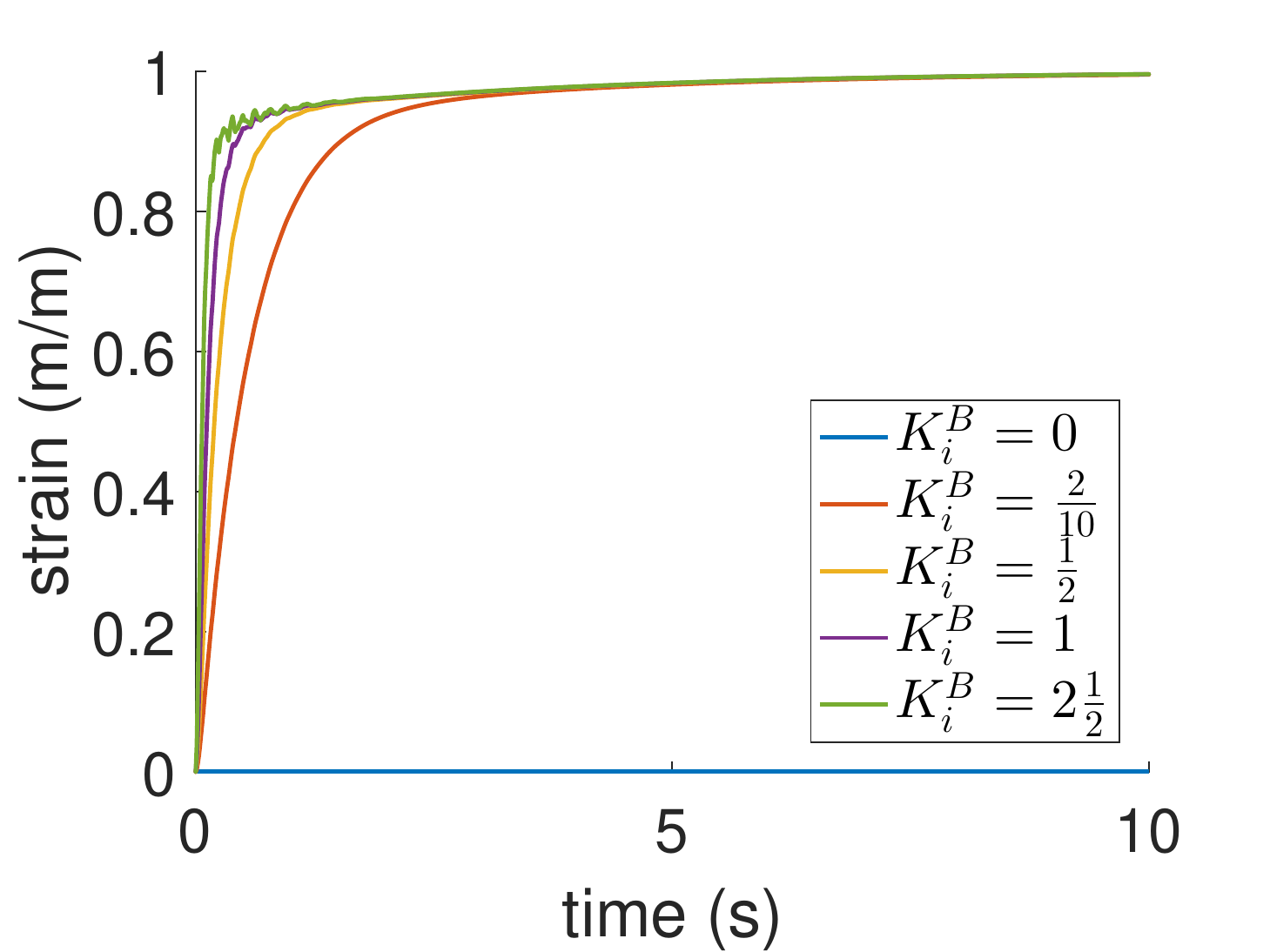}
					\caption{\tiny{$K_p^B=0$, $\bar{U}^*=\tfrac{3}{4}$}}
				\end{subfigure}
				\begin{subfigure}[b]{0.32\columnwidth}
					\includegraphics[width=\columnwidth]{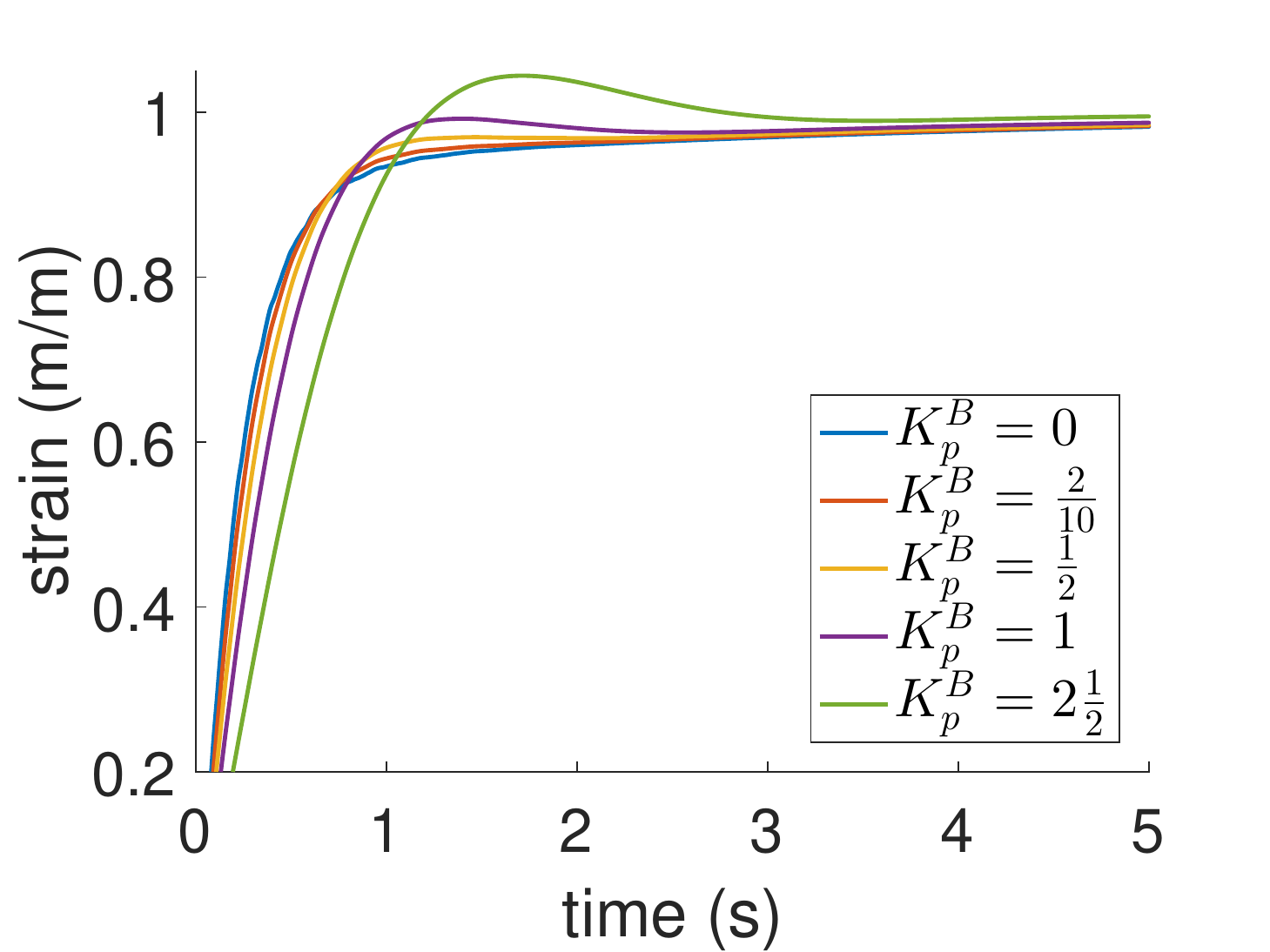}
					\caption{\tiny{$K_i^B=\tfrac{1}{2}$, $\bar{U}^*=0$}}
				\end{subfigure}
				\begin{subfigure}[b]{0.32\columnwidth}
					\includegraphics[width=\columnwidth]{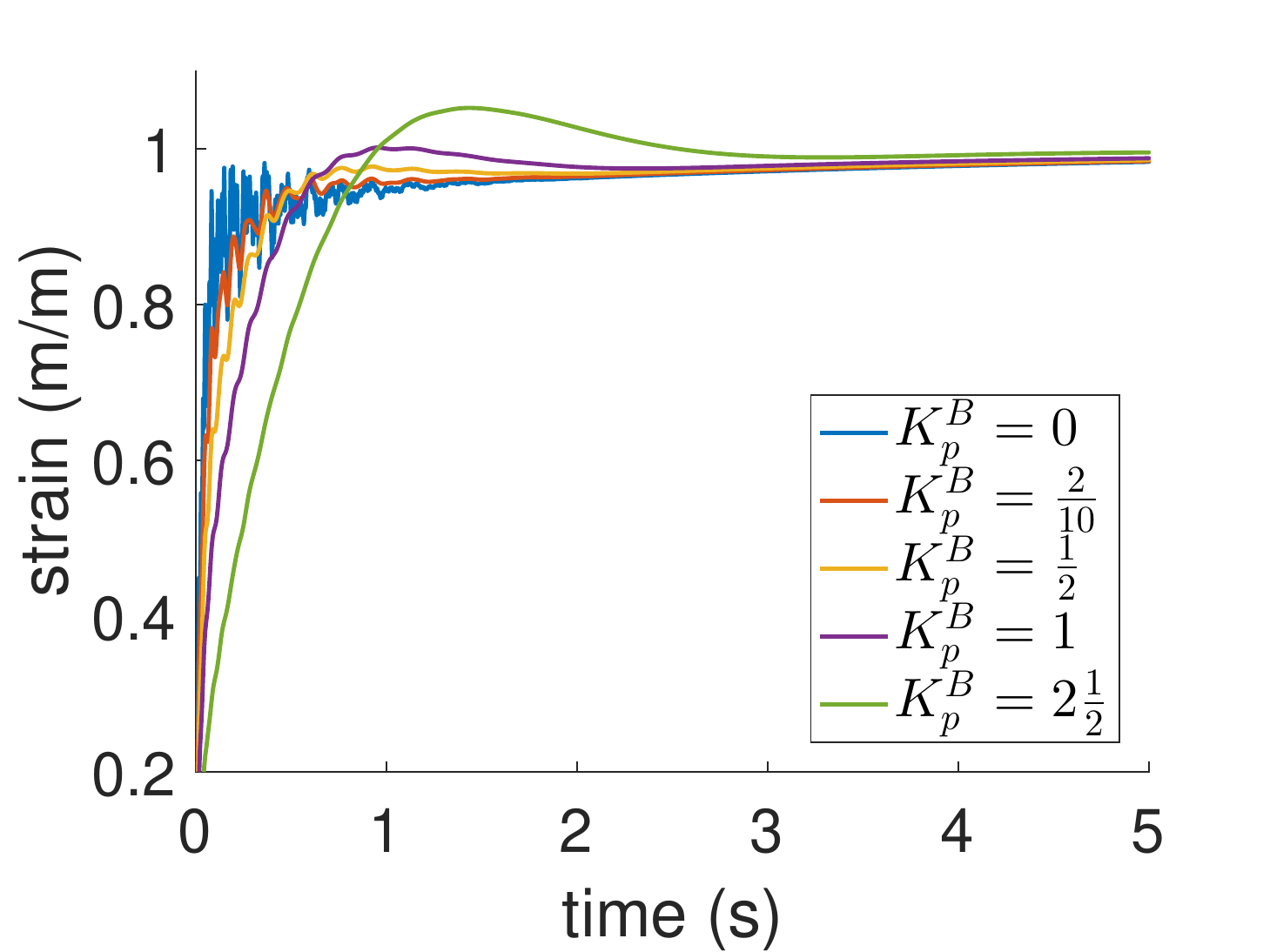}
					\caption{\tiny{$K_i^B=\tfrac{1}{2}$, $\bar{U}^*=\tfrac{3}{4}$}}
				\end{subfigure}
				\caption{Closed-loop response through input shaping. - In (a) the integral action is shown. In (b) the proportional action varies for a fixed integral action, without reference action. In (c) same as in (b) except for inclusion of the reference action. The used parameters are given in Table \ref{tab:params_method_A}.} 
				\label{fig:sim_method_B}
			\end{figure}
			Next, we start with showing the first condition \eqref{stab_th_firstcond} in Theorem \ref{thm::stab}. Firstly, it is straight forward to show that $\delta_{x}S_d(x^\ast)=0$. Secondly, consider the time derivative of the closed-loop storage function, along the closed-loop trajectories of the piezoelectric beam \eqref{eq:model_efforts} with controller \eqref{eq:stabalizing_controller_B}
			\begin{eqnarray}\label{eq_Sddot_a}
			\begin{split}
			\dfrac{\partial }{\partial t}S_d
			&= -\int_{0}^{1}b\de_p^2dz-K_p\dot{e}_p^2(1)\leq 0,
			\end{split}
			\end{eqnarray}	
			where, in the first line we made use of Proposition \ref{prop:passivity} and in the third line, we made use of equation \eqref{os_diff_contr}. This implies that $ \frac{\partial }{\partial t}S_d\leq 0$. Denoting $\Delta x= x-x^\ast$; $\mathcal{N}(\Delta x)= S_d(x^\ast +\Delta x)-S_d(x^\ast)$ can be simplified as 
			\begin{eqnarray*}
				\mathcal{N}(\Delta x)&=\tfrac{1}{4}\int_{0}^{1}\left(\tfrac{1}{C}\left(\Delta\dot{ e}_q\right)^2+\rho \left(\Delta\dot{e}_p\right)^2+C\left(d\Delta e_p\right)^2 \right.\\&\left.+\dfrac{1}{\rho} \left(d\Delta e_q-b\Delta e_p\right)^2\right)dz+\frac{K_i}{2}\braces{\Delta e_p(1)}^2,
				\end{eqnarray*}
				where we made use of $S$ in Remark \ref{rem_storage_edot_and_e} and \eqref{eq:storagefunc_planta}.
			Then, the condition given in equation \eqref{stabthm_cond2} of Theorem \ref{thm::stab} holds with
			\begin{eqnarray*}
				||\Delta x||^2&=\int_{0}^{1}\left(\left(\Delta\dot{ e}_q\right)^2+ \left(\Delta\dot{e}_p\right)^2+\left(d\Delta e_p\right)^2 \right.\\&\left.+ \left(d\Delta e_q-b\Delta e_p\right)^2\right)dz+\braces{\Delta e_p(1)}^2,
			\end{eqnarray*} $\alpha=2$, $\gamma_1=\frac{1}{4}\min\{\frac{1}{C},\rho,2K_i\}$ and $\gamma_2=\frac{1}{4}\max\{\frac{1}{C},\rho,2K_i\}$.
		\end{proof}
		\subsection{Passivity based control: Input shaping }\label{sec:input_shaping}
		Besides shaping the output port-variable, the input port-variable can also be shaped, i.e. design a controller for $\dot{U}(t)$. In doing so, a new technique is presented, called {\em input shaping}. This methodology uses the integrability property of the input port-variable $\dot{U}$ (rather than the output port-variable) to shape the closed-loop storage function. The authors believe that this has never been explored in the literature for infinite-dimensional systems and is therefore considered to be a novel contribution of this paper. A comparison between output shaping and input shaping techniques is presented in the next section.
		\begin{prop}
			Let Assumption \ref{ass:existance} and \ref{ass:available} hold and consider the piezoelectric beam \eqref{eq:model_efforts}, controlled through the boundary by
			\begin{eqnarray}\label{eq:stabalizing_controller_input_shaping}
			\begin{split}
			\dot \psi&= \left(U-U^\ast\right)\\
			U(t)&=
			-\frac{1}{K_p}\brackets{K_i\psi-\gamma\braces{e_p(1)-e_p^*}}+U^\ast,
			\end{split}
			\end{eqnarray}
			where $K_i,~K_p>0$ are tuning parameters, $\psi\in \mathbb{R}$. Then, the closed-loop system is stable at the operating point \eqref{eq:desired_equilibrium}.
		\end{prop}
	\begin{proof}
		By using the storage function $$	S_d:=S+\frac{1}{2}K_i\braces{U-U^\ast}^2,$$ the proof is analogous to that of Proposition \ref{prop:controla}. 
	\end{proof}

		\section{Analysis and simulation results}
		In this section the proposed controllers are analyzed and simulation results are presented. The closed-loop systems are implemented using the mixed-finite-element method \cite{GoloSchaft2004} with $N=16$ segments. The time-discretization is performed with a RK4 scheme \cite{bradie2006friendly}, using the fixed time-step $\Delta t=0.0001$. 
		The assigned parameters of the open-loop system \eqref{eq:ph_beam_eqn} are of illustrative nature and take the values $(\rho,\ \ell,\ C,\ \gamma,\ b)=(
		1,\ 1,\ \tfrac{3}{4},\ \tfrac{1}{10},\ 7)$. It is possible to use parameters that are nigh to true piezoelectric material. However, the numerics become more tedious. Using the initial condition as the open-loop equilibrium $e^q(z)=e^p(z)=0$ $\forall z$. The control objective is given in \eqref{eq:desired_equilibrium} to stabilize the closed-loop systems at the desired equilibrium \eqref{eq:equilibrium_poi}, with desired strain $\mathcal{E}^*=1$. The initial conditions of the integrator action of the controllers are chosen to be zero (since \eqref{eq:stabalizing_controller_B} and \eqref{eq:stabalizing_controller_input_shaping} already take the desired equilibrium \eqref{eq:desired_equilibrium} into account). The simulations provide insights in the influence of parameter tuning on the closed-loop behaviour. This allows for a comparison of the two controllers.
	 \subsection{Analysis \& simulation: Method output shaping}   For the purpose of comparing the controllers from sections \ref{sec:output_shaping} \textit{(output shaping)} and \ref{sec:input_shaping} \textit{(input shaping)}, it is useful to define $\bar{U}(t):=-\gamma U(t)$ such that $\bar{U}^*={-\gamma U^*}$, then \eqref{eq:stabalizing_controller_B} can be written as
		\begin{eqnarray}\label{eq:stabalizing_controller_A_adjusted}
			\begin{split}
			\dot{\phi}&= e_p(1,t)-e_p^*\\
			\bar{U}(t)&=
			{-K^A_i\phi-K^A_p\braces{e_p(1,t)-e_p^*}}+\bar{U}^\ast,
			\end{split}
			\end{eqnarray}
			where $K^A_p:={K_p}$ and $K^A_i:={K_i}$. The controller  \eqref{eq:stabalizing_controller_A_adjusted} can be interpreted as a composition of three super-imposed control actions. In particular, an integral action with respect to the output of the open-loop system, effectuated by $K_i^A$, a proportional action with respect to the output of the open-loop system, effectuated by $K_p^A$, and a reference action with magnitude $\bar{U}^*$ for the input $\bar{U}(t)$.\\
		Three situations are considered. The tuning of the proportional action and integral action can be seen in Fig \ref{fig:sim_method_A}(a) and Fig \ref{fig:sim_method_A}(b), respectively. The influence of the reference action is shown in Fig \ref{fig:sim_method_A}(c). Not considering a reference action result in a 70\% steady-state error. An overview of the used variables is given in Table \ref{tab:params_method_A}.\\
		    In Fig \ref{fig:sim_method_A}(a) it can be seen that the proportional action damps the response. By further increase of the proportional action, the system becomes over-damped. Note that the early vibrations for $K_p^A=0$ are diminished by the proportional action. In Fig \ref{fig:sim_method_A}(b) it can be seen that more integral action result in higher overshoot and more vibrations. However, a subtle integral action could benefit the rise-time of the closed-loop system. In Fig \ref{fig:sim_method_A}(c) it can be seen that the reference action influences the steady-state. In Fig \ref{fig:sim_method_A}(d) the steady-state error goes to zero. In both Fig \ref{fig:sim_method_A}(c-d) it is shown that for a fixed integral action, the proportional action can be used to damp-out vibrations and reduce the overshoot, from the integral action.  
        \begin{remark}
    		In this note, the output of the open-loop system is the velocity at the tip, i.e. $e_p(1,t)=\dot{w}(1,t)$, The controller proposed in \cite{Macchelli2017_BCLDPS_dissipation} is described as a PD-controller with respect to the longitudinal displacement at the tip, i.e. $w(1,t)$. Therefore, it can be concluded that their proportional action coincides with the integral action of the output shaping controller \eqref{eq:stabalizing_controller_A_adjusted}, and their derivative action coincides with the proportional action of the output shaping controller \eqref{eq:stabalizing_controller_A_adjusted}. Moreover, if a step-response is applied to the closed-loop system proposed in \cite{Macchelli2017_BCLDPS_dissipation}, the dynamical controller \eqref{eq:stabalizing_controller_A_adjusted} and the one proposed in \cite{Macchelli2017_BCLDPS_dissipation} coincide. The storage function proposed in this paper takes on a differential passivity approach, while the storage function used in \cite{Macchelli2017_BCLDPS_dissipation} to stabilize the system, takes on an incremental passivity approach.
		\end{remark}
		\begin{figure}
			\noindent
			\begin{minipage}[t]{0.25\textwidth}
				\raggedright
				\includegraphics[width=\textwidth]{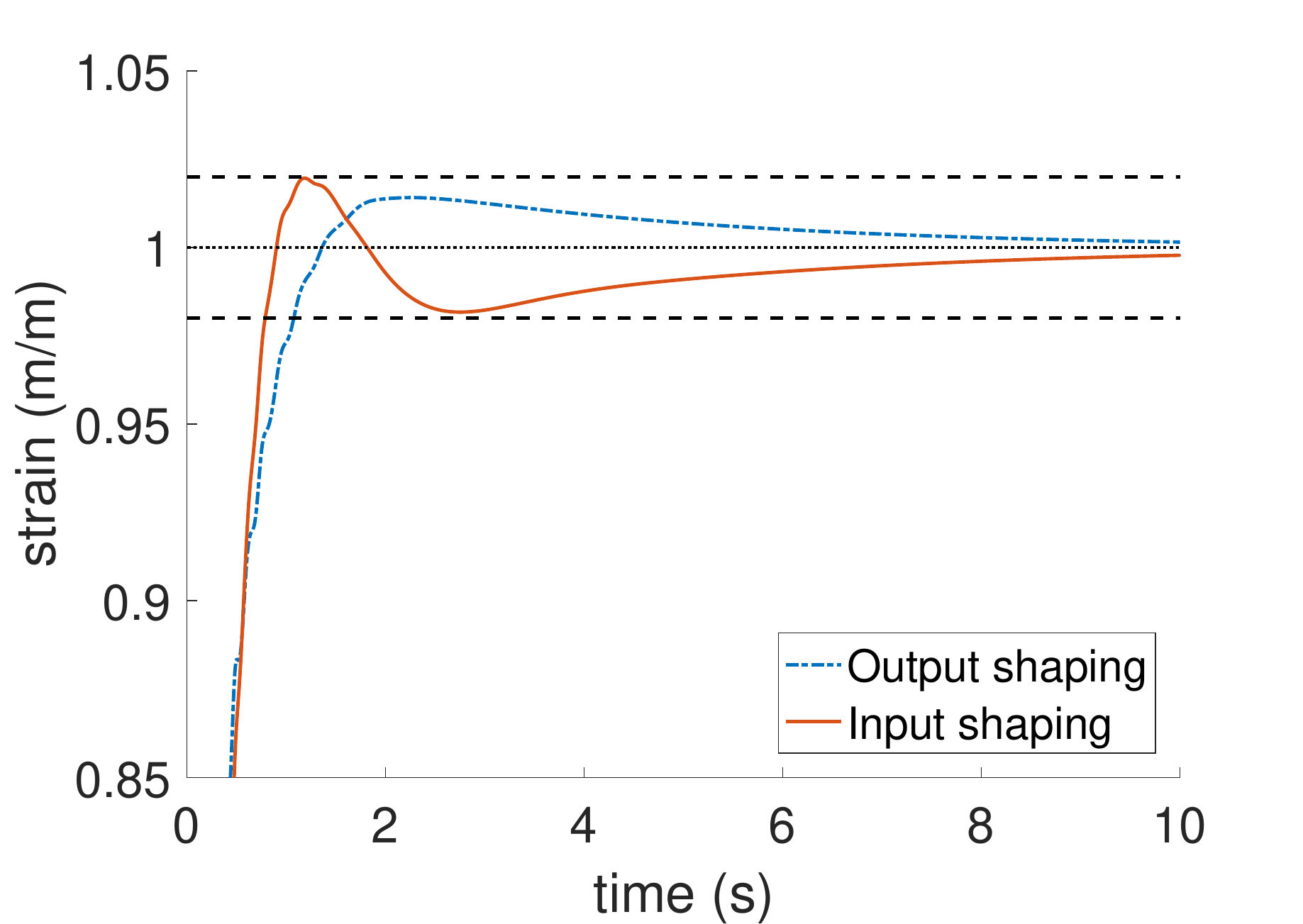}
			\end{minipage}%
			\begin{minipage}[t]{0.25\textwidth}
				\includegraphics[width=\textwidth]{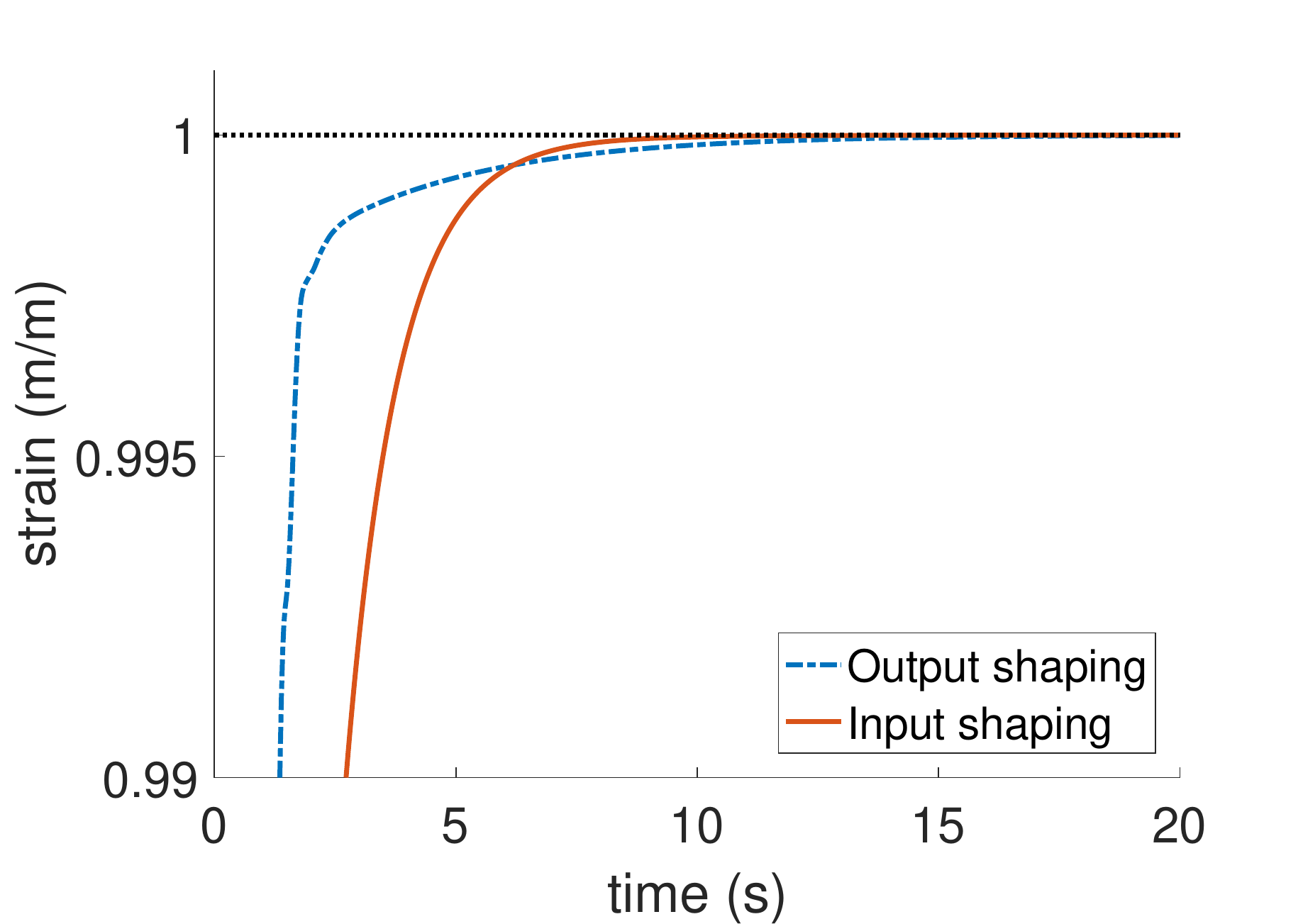}
			\end{minipage}
					\caption{Closed-loop response for tuned controllers - Left: case 1. Controller A has gains $(K_i^A,\ K_p^A) =(0.5,\ 1.35)$, controller B has gains $(K_i^B,\ K_p^B)=(0.45,\ 1.4)$. Right: case 2. Controller A has gains $(K_i^A,\ K_p^A) =(0.25,\ 0.65)$ and controller B has gains $(K_i^B,\ K_p^B)=(0.1,\ 0.5)$.} 
		\label{fig:sim_method_AB}
		\end{figure}
	 \subsection{Analysis  \& simulation: Input shaping} 
	 
	 Again, let $\bar{U}(t)=-\gamma U^*$, as with \eqref{eq:stabalizing_controller_B}. Then,  \eqref{eq:stabalizing_controller_input_shaping} becomes
		\begin{eqnarray}\label{eq:stabalizing_controller_input_shaping_adjusted}
			\begin{split}
			\dot \psi&= \left(U-U^\ast\right)\\
			\bar{U}(t)&=
			{K^B_i\psi-K^B_p\braces{e_p(1)-e_p^*}}+\bar{U}^\ast,
			\end{split}
			\end{eqnarray}
		with tuning gains $K^B_p:=\frac{\gamma^2}{K_p}$ and $K^B_i:=\frac{\gamma K_i}{K_p}$.
 		The controller \eqref{eq:stabalizing_controller_input_shaping_adjusted} can be interpreted as a composition of three super-imposed actions, i.e. an integral action with respect to the input $U(t)$, a proportional action with respect to the output of the open-loop system, and a reference action with magnitude $\bar{U}^*$ for the input $\bar{U}(t)$.
\begin{table}
			\caption{Used variables - Output shaping} \label{tab:params_method_A} 
			\begin{center}
				\begin{tabular}{ l | l | l | c }
					\hline
					Fig & Fixed & Variable & Reference \\ \hline \hline
					 \ref{fig:sim_method_A}(a) & $K_i^A=0$ & $K_p^A\in\{0 , \frac{1}{2} ,  2 ,5\}$ & \cmark\\  \hline
					 \ref{fig:sim_method_A}(b) & $K_p^A=0$ & $K_i^A\in\{0 , \frac{1}{10} ,\frac{3}{10} ,\frac{1}{2} \}$ & \cmark \\ \hline
					 \ref{fig:sim_method_A}(c) & $K_i^A=\tfrac{3}{10}$ &  $K_p^A\in\{0 ,  \frac{2}{10} ,\frac{1}{2} , 1 , 2\frac{1}{2} \}$ & \cmark\\ \hline \hline
					 \ref{fig:sim_method_B}(a) 	& $K_p^B=0$ 				& $K_i^B\in\{0 , \frac{2}{10} , \frac{1}{2} , 1 , 2\frac{1}{2} \}$						 								& \xmark \\ \hline
					 \ref{fig:sim_method_B}(b) 	& $K_i^B=\tfrac{1}{2}$    &  $K_p^B\in\{0 ,  \frac{2}{10} ,\frac{1}{2} , 1 , 2\frac{1}{2} \}$ 										& \xmark\\  \hline 
					 \ref{fig:sim_method_B}(c) 	& $K_i^B=\tfrac{1}{2}$ 	&  $K_p^B\in\{0 , \frac{2}{10} ,\frac{1}{2} , 1 , 2\frac{1}{2} \}$										& \cmark\\ \hline
				\end{tabular}
			\end{center}
		\end{table}	
		Four situations are considered for the stabilizing controller \eqref{eq:stabalizing_controller_input_shaping_adjusted}. The influence of the proportional action coincide with the behaviour \ref{fig:sim_method_A}(b)
		and integral action can be seen in Fig \ref{fig:sim_method_B}(a). Simple calculations show that if only the integral action and the reference action are considered, the integral action has no effect since  $\bar{U}(t)-\bar{U}^*=U(t)-U^*=0$. Therefore, in Fig \ref{fig:sim_method_B}(a) is the reference action excluded. Note that if all three actions are active, then the integral action acts on the changes in $\bar{U}(t)$ caused by the proportional action. The influence of the reference action is shown in Fig \ref{fig:sim_method_B}(b-c). An overview of the used variables is given in Table \ref{tab:params_method_A}.\\
		In Fig \ref{fig:sim_method_B}(a) it is shown that increasing the integral action result in a rise-time reduction at the expense of vibrations. Comparing Fig \ref{fig:sim_method_B}(b) with \ref{fig:sim_method_B}(c) shows that the reference action increases rise-time at a cost of vibrations. The vibrations can be damped-out by means of the proportional action at a cost of overshoot and settling-time.
		
\subsection{Comparison: Output shaping \& Input shaping}
		Finally, the output and input shaping controllers, respectively \eqref{eq:stabalizing_controller_A_adjusted} and \eqref{eq:stabalizing_controller_input_shaping_adjusted}, are compared. It can be seen in \eqref{eq:stabalizing_controller_A_adjusted} and \eqref{eq:stabalizing_controller_input_shaping_adjusted} that the integral actions are different, whereas the proportional and reference actions coincide. The proportional actions of both \eqref{eq:stabalizing_controller_A_adjusted} and \eqref{eq:stabalizing_controller_input_shaping_adjusted} can be seen as damping injection and the integral action of \eqref{eq:stabalizing_controller_A_adjusted} shapes the energy.

		The performance of the controllers are tested for the following two requirement cases 
        \begin{enumerate}
            \item minimizing the settling-time for a $2\%$ margin, a maximum overshoot of $2\%$, limited vibrations, and zero steady-state error.
            \item  Minimizing the settling-time for a $2\%$ margin, a maximum overshoot of $0\%$, limited vibrations, and zero steady-state error.
        \end{enumerate}
        The results for case 1 and 2 are respectively depicted in Fig \ref{fig:sim_method_AB}(a) and \ref{fig:sim_method_AB}(b).
		In Fig \ref{fig:sim_method_AB}(a) it can be seen that the rise-times are similar. However, the controller for input shaping has a faster settling-time. In Fig \ref{fig:sim_method_AB}(b) it can be seen that the oscillations due to the input shaping requires more proportional action, which comes at a cost of the rise-time. However, it can also be observed that the controller through input shaping reaches high accuracy faster than the controller obtained through output shaping. Both methods result in smooth transients.

		
		\section{Conclusions }
		In this paper we have presented a new passive map for piezoelectric systems. The port-variables of this passive map are integrable. Next, we have considered the boundary control problem of this system. To achieve this we have utilized the integrability property of the port-variable and present two boundary control techniques for regulating the strain. An important advantage of this techniques is that there is no need to solve for Casimir functions, which is generally the case \cite{915398}. 

		\bibliographystyle{IEEEtran}
		\bibliography{references}
	\end{document}